\documentclass{revtex4}
\usepackage{amsmath,amsthm}
\usepackage{rotating}
\usepackage{graphicx}

\usepackage{amsfonts}
\newtheorem{theorem}{Theorem}[section]
\newtheorem{lemma}[theorem]{Lemma}
\newtheorem{proposition}[theorem]{Proposition}
\newtheorem{cor}[theorem]{Corollary}

\theoremstyle{remark}
\newtheorem{remark}[theorem]{Remark}

\theoremstyle{definition}
\newtheorem{definition}[theorem]{Definition}

\theoremstyle{example}
\newtheorem{example}[theorem]{Example}

\theoremstyle{notation}

\newcommand{\bra}[1]{\langle#1|}
\newcommand{\ket}[1]{|#1\rangle}
\begin{document}

\title{Multipartite quantum systems: an approach based on Markov matrices and the Gini index}            
\author{A. Vourdas}
\affiliation{Department of Computer Science,\\
University of Bradford, \\
Bradford BD7 1DP, United Kingdom\\a.vourdas@bradford.ac.uk}

\begin{abstract}
An expansion of row Markov matrices in terms of matrices related to permutations with repetitions, is introduced.
It generalises the Birkhoff-von Neumann expansion of doubly stochastic matrices in terms of permutation matrices (without repetitions).
An interpretation of the formalism in terms of sequences of integers that open random safes
described by the Markov matrices, is presented.
Various quantities that describe probabilities and correlations in this context, are discussed.
The Gini index is used to quantify the sparsity (certainty) of various probability vectors.
The formalism is used in the context of multipartite quantum systems with finite dimensional Hilbert space,
which can be viewed as quantum permutations with repetitions or as quantum safes.
The scalar product of row Markov matrices, the various Gini indices, etc, are novel probabilistic quantities that describe the statistics of multipartite quantum systems.
Local and global Fourier transforms are used to define locally dual and also globally dual statistical quantities.
The latter depend on off-diagonal elements that entangle (in general) the various components of the system.
Examples which demonstrate these ideas are also presented.

\end{abstract}
\maketitle

\section{Introduction}

Multipartite quantum systems play an important role in quantum mechanics and quantum information.
They are associated with deep quantum concepts like entanglement, correlations which are stronger than classical, violation of classical probabilistic inequalities, etc.
Various statistical quantities are used to describe these phenomena.
In this paper we combine ideas from Markov matrices, permutations with repetitions, and the Gini index in statistics, to develop novel tools for their study.
We first develop these ideas at a classical probabilistic level, and then pass to the quantum level.

Markov matrices play an important role in Markov chains \cite{M1,M2}, Artificial Intelligence, Engineering, etc.
A special case of Markov matrices are the doubly stochastic matrices which are intimately related to majorization\cite{MAJ} and have many applications.
An important theorem is the Birkhoff-von Neumann expansion of doubly stochastic matrices in terms of permutation matrices (without repetitions), which has many applications in Operational Research, in 
Linear Programming, etc. In this paper we introduce analogous expansions for row Markov matrices, in terms of matrices related to permutations with repetitions of $d$ integers in ${\mathbb Z}_d$.

A nice physical interpretation of this mathematical formalism is given, in terms of  sequences of integers that open random safes described by the Markov matrices.
Various statistical quantities like joint probabilities and correlations, are presented in this context.

Probability vectors multiplied by row Markov matrices on the right, are transformed into other probability vectors.
The sparsity of probability vectors is studied using  Lorenz values and the Gini index\cite{Gini,Gini1,Gini2,Gini3,Gini4} which have been used extensively in Mathematical Economics.
We have recently used these concepts in the context of quantum physics\cite{V2,V3}.
In the context of random safes we introduce the Gini vector that describes the sparsity 
in the local probability vector for  each of the  integers in the sequence that opens a random safe. We also introduce the total Gini index that describes the sparsity of the 
joint probabilities for the sequence of integers that opens  a random safe.

This formalism is applied to $d$-partite quantum systems where each component is described by a $d$-dimensional Hilbert space.
There is a natural link between such
quantum systems and permutations with repetitions of $d$ integers in ${\mathbb Z}_d$ (compare Eq.(\ref{A7}) with Eqs(\ref{333a}),(\ref{333})).
The present work enhances this link, and
the outcome is a plethora of statistical quantities that describe various aspects of multipartite quantum systems.
Local Fourier transforms are used to define locally dual statistical quantities.
Global Fourier transforms are used to define globally dual statistical quantities which
 depend on off-diagonal elements that entangle (in general) the various components of the system.

In section 2 we define permutations with repetitions.
In section 3 we introduce Markov matrices and explain their role in describing the statistics of random permutations with repetitions.
In section 4, we introduce an expansion for row Markov matrices, in terms of matrices related to permutations with repetitions.
An important concept here is the scalar product of two row Markov matrices and its physical interpretation.

In section 5 we introduce the Lorenz values of probability vectors, and the related topic of majorization. In the present context majorization
provides an ordinal description of the sparsity in probability vectors.

In section 6 Lorenz values are used to define the Gini index as an indicator of the sparsity (certainty) of probability vectors.
Both the Gini index and the variance (standard deviation) are quantities that describe the variability in probability distributions. 
Roughly speaking, the Gini index is the sum of the absolute values of the differences of the probabilities, while the variance 
is the sum of the squares of the differences of the probabilities (remark \ref{rem1}).
In the context of random safes we introduce local and total Gini indices.

In section 7 we consider a quantum system with $d$-dimensional Hilbert space $H_d$, and variables in ${\mathbb Z}_d$.
We have shown in ref\cite{V2} that the Gini indices for 
probability vectors related to positions and momenta (i.e., dual bases through a Fourier transform) cannot both be very sparse
and we quantified this using Gini indices. It is one way to express the uncertainty principle in this context.
Here we simply state this result without proof because we generalise it later for multipartite systems.

In section 8 we consider $d$-partite quantum systems where each component is described by the $d$-dimensional Hilbert space $H_d$.
They can be viewed as quantum permutations with repetitions (or as quantum safes)
and they are extensions of their random counterparts, because they involve superpositions.
This point of view brings the statistical techniques related to
row Markov matrices, Gini indices, etc, into multipartite quantum systems.
It leads to many novel quantities that describe the statistics of quantum systems.

In quantum mechanics it is interesting to do statistics with respect to dual bases related through a Fourier transform, and study uncertainty relations that link them.
In multipartite systems we have local and global Fourier transforms which are used in sections 8,9, to define locally dual and also globally dual statistical quantities.
We conclude in section 10 with a discussion of our results.

\section{Generalisations of permutations}
\subsection{Permutations}
Let ${\mathbb Z}_d$ be the ring of integers modulo $d$, and  $\pi$ a permutation of its elements:
\begin{eqnarray}
(0,1,...,d-1)\;\overset{\pi} \rightarrow\;(\pi(0),\pi(1),...,\pi(d-1)).
\end{eqnarray}
$\pi$ is a bijective map from ${\mathbb Z}_d$ to ${\mathbb Z}_d$, and is an element of the symmetric group ${\cal S}$  which has $d!$ elements.
Multiplication in this group is the composition:
\begin{eqnarray}
(0,1,...,d-1)\;\overset{\pi} \rightarrow\;(\pi(0),\pi(1),...,\pi(d-1))\overset{\wp} \rightarrow\;(\wp [\pi(0)],\wp[\pi(1)],...,\wp[\pi(d-1)]).
\end{eqnarray}
We will use the notation $(\wp \circ \pi)(i)=\wp[\pi(i)]$.
The unity is
\begin{eqnarray}
(0,1,...,d-1)\;\overset{\bf 1} \rightarrow\;(0,1,...,d-1).
\end{eqnarray}

A permutation matrix $M_\pi$ is a $d\times d$ matrix with elements 
\begin{eqnarray}\label{PP}
M_\pi(i,j)=\delta(\pi(i),j);\;\;\;{\rm rank} (M_\pi)=d.
\end{eqnarray}
where $\delta$ is the Kronecker delta. 
Each row and each column have one element equal to $1$ and the other $d-1$ elements equal to $0$.
Clearly
\begin{eqnarray}
M_\wp M_\pi=M_{\pi\circ \wp};\;\;\;M_{\bf 1}={\bf 1};\;\;\;M_{\pi^{-1}}=[M_\pi]^{-1}=[M_\pi]^T.
\end{eqnarray}
The matrices $M_\pi$ form a representation of the symmetric group ${\cal S}$\cite{SYM}.

\subsection{Permutations  with repetitions}

We enlarge the set of permutations by considering maps which might not be bijective. 
This is equivalent to choosing $d$ integers from ${\mathbb Z}_d$, allowing repetitions.
Let ${\cal F}$ be the set of all functions
\begin{eqnarray}
f:\;(0,...,d-1)\;\rightarrow\;(f(0),...,f(d-1));\;\;\;f(i)\in{\mathbb Z}_d,
\end{eqnarray}
from ${\mathbb Z}_d$ into itself.
An example of the function $f(i)$ is to consider a safe that opens with a sequence of $d$ integers in ${\mathbb Z}_d$ that has in the $i$-position the number $f(i)$.
This example is discussed further below.

There are $d^d$ functions in ${\cal F}$ (because $0$ can be mapped to any of the $d$ elements in ${\mathbb Z}_d$, $1$ can be mapped to any of the $d$ elements in ${\mathbb Z}_d$, etc).
From them $d!$ are permutations.
The set ${\cal F}$ with the composition 
\begin{eqnarray}
(f\circ g) (i)=f[g(i)]
\end{eqnarray}
as multiplication, is a semigroup (because the inverse does not exist in general).
In terms of the example with a safe, the composition represents a change in the sequence that opens the safe.

The function $f\in{\cal F}$ can be represented with the matrix $M_ f$ with elements  
\begin{eqnarray}\label{6}
M_f(i,j)=\delta(f(i),j);\;\;\;{\rm rank} (M_f)\le d.
\end{eqnarray}
This is a generalisation of Eq.(\ref{PP}), and includes functions which might not be permutations.
Each row has one element equal to $1$ and the other $d-1$ elements equal to $0$.
But a column might have many elements equal to one.
There are $d^d$ matrices $M_f$ and we refer to them as permutation with repetition matrices. $d!$ of them are permutation matrices.

We note that
\begin{eqnarray}
M_gM_f=M_{f\circ g},
\end{eqnarray}
where $f\circ g$ is the composition of the two functions.
Therefore the set of these matrices with matrix multiplication, is a semigroup isomorphic to ${\cal F}$, and we denote it as ${\cal F}$.

We also consider transpose of these matrices (which have the $(j,i)$ element equal to $\delta(f(i),j)$) that form another semigroup isomorphic to ${\cal F}$.
Here each column has one element equal to $1$ and the other $d-1$ elements equal to $0$.
But a row might have many elements equal to one.

\begin{lemma}\label{LL}
\begin{eqnarray}
\prod _{i=0}^{d-1}[M_fM_g^T](i,i)=\delta(f,g);\;\;\;M_f, M_g \in {\cal F}_R.
\end{eqnarray}
\end{lemma}
\begin{proof}
The $(i,i)$-element of the matrix $M_fM_g^T$ is
\begin{eqnarray}
[M_fM_g^T](i,i)=\sum _j\delta(f(i),j)\delta(g(i),j)
\end{eqnarray}
If $f\ne g$ at least one of these $(i,i)$-elements is $0$, and then their product is $0$.
Only in the case $f=g$ all $(i,i)$-elements are $1$, and then their product is $1$.
\end{proof}
Below we will see that when the matrices $M_f$ (and the row Markov matrices which are more general) act on the right of probability vectors (written as rows) we get 
other probability vectors.

\section{Random permutations with repetitions: random safes}

\subsection{Markov matrices }

A $d\times d$ matrix $q$ is called row Markov matrix if 
\begin{eqnarray}\label{A7}
\sum _{j=0}^{d-1}q(i,j)=1;\;\;\;q(i,j)\in [0,1];\;\;\;i,j=0,...,d-1.
\end{eqnarray}
The elements of each  row are probabilities and we
interpret them in terms of a random safe as follows.
We consider a `large' ensemble of safes each of which opens with a sequence of $d$ integers which take values in ${\mathbb Z}_d$.
Let $q(i,j)$ be the probability that the number in the $i$-position of the opening sequence, is $j$. 
The matrix $q$ is a row Markov matrix, because the sum of probabilities that in the $i$-position is some  number in ${\mathbb Z}_d$, is $1$.
We call the index $i$ `position index', and the index $j$ `number index'.

We denote as ${\cal M}$ the set of row Markov matrices.
We note that
\begin{itemize}
\item
The product of two row Markov matrices is a  row Markov matrix.
\item
The inverse of a row Markov matrix might not exist, or if it exists it might not be a row Markov matrix.
\item
${\cal M}$ is a  semigroup with respect to matrix multiplication.

\item
If $q_1, q_2$ are row Markov matrices, then the $\lambda q_1+(1-\lambda) q_2$ where $0\le \lambda \le 1$ is a row Markov matrix.
The semigroup ${\cal M}$ of all $d\times d$ row Markov matrices is a convex polytope with the $d^d$ matrices $M_f$ (in the semigroup ${\cal F}$) as vertices.
This polytope is studied and used in this paper.

\item
A special case of row Markov matrices are the doubly stochastic matrices for which the following relation also holds:
\begin{eqnarray}
\sum _{i=0}^{d-1}q(i,j)=1.
\end{eqnarray}
They are intimately related to majorization\cite{MAJ}, which is a preorder that has been used in various areas, including quantum physics\cite{MA0,MA1,MA2,MA3}.

\item
The matrix ${\mathfrak U}$ with all elements
\begin{eqnarray}\label{UUU}
{\mathfrak U}_{ij}=\frac{1}{d};\;\;\;{\rm rank}({\mathfrak U})=1,
\end{eqnarray}
is a doubly stochastic matrix. 
For any row Markov matrix $q$ we get 
\begin{eqnarray}
q{\mathfrak U}={\mathfrak U},
\end{eqnarray}
${\mathfrak U}q$ is a row Markov matrix with all rows equal to each other, but in general ${\mathfrak U}q\ne{\mathfrak U}$.
\end{itemize}

Let ${\cal P}$ be the set of probability vectors  written as rows
\begin{eqnarray}\label{gt}
{\bf x}=(x(0),...,x(d-1));\;\;\;\sum _i x(i)=1;\;\;\;x(i)\ge 0.
\end{eqnarray}
If $q$ is a row Markov matrix then 
${\bf x}q$ is a probability vector, but $q{\bf x}^T$ might not be a probability vector.
Examples of probability vectors are the `most uncertain probability vector' ${\bf u}$, and the `certain probability vectors' ${\bf c}M_\pi$, where
\begin{eqnarray}\label{3D}
{\bf u}=\frac{1}{d}(1,...,1);\;\;\;{\bf c}=(1,0...,0).
\end{eqnarray}

A probability vector is called sparse (`almost certain') if most of its elements are zero or almost zero. 
Roughly speaking, sparse probability vectors are at the opposite end of the uncertain probability vector ${\bf u}$.
Below we make this precise with the preorder majorization, and also with the Gini index.

The Birkhoff-von Neumann theorem states that every doubly stochastic matrix   ${\mathfrak D}$ can be expanded in terms of the $d!$ permutation matrices as
\begin{eqnarray}\label{hh}
{\mathfrak D}=\sum _{\pi}\lambda _{\mathfrak D}({\pi})M_{\pi};\;\;\;\sum _{\pi} \lambda _{\mathfrak D}({\pi})=1;\;\;\; \lambda _{\mathfrak D}({\pi})\ge 0;\;\;\;\pi \in {\cal S}.
\end{eqnarray}
Here $\lambda _{\mathfrak D}({\pi})$ are probabilities. 
So a doubly stochastic matrix can be viewed as a random permutation (without repetitions).
In this paper we introduce expansions analogous to this for row Markov matrices, that involve permutations with repetitions.
They can be viewed as random permutations with repetitions.
Later we go one step further, and introduce quantum  permutations with repetitions, that involve superpositions.

\section{Expansions for row Markov matrices and their interpretation in terms of random safes}

In this section we introduce expansions for row Markov matrices, in terms of matrices related to permutations with repetitions of $d$ integers in ${\mathbb Z}_d$.
These expansions are interpreted in terms of the sequences of integers that open random safes described by the Markov matrices.

The  expansion is not unique, and it depends on the correlations between the integers in the sequence that opens the random safe.
 We first assume absence of correlations and then joint probabilities are equal to products of probabilities.
 We show that the coefficients in the expansion are scalar products of row Markov matrices, a concept which is defined and plays an important role in this paper.
Later we consider the case of correlations.

The Markov matrix $q(i,j)$ describes the probability that in the position $i$ of the sequence that opens the random safe, is the integer $j$.
It does not give any information about joint probabilities and correlations. 
We can have two different ensembles described by the same Markov matrix $q(i,j)$, one without any correlations between the integers in the sequence that opens the random safe, and the other with correlations.
In this case we have two different expansions of the same row Markov matrix.

\subsection{Expansion for row Markov matrices in the absence of correlations}\label{uncor}
 In this subsection we assume independence (lack of correlations) between the integers in the various positions of the sequences that open the safes.
 We consider two ensembles of safes (which are independent of each other), described by the row Markov matrices $q$, $p$.
 We take randomly one safe from each ensemble, and then
the product $q(i,j)p(i,j)$ is the joint probability that in both safes the number in the $i$-position of the opening sequence is $j$.
Summation over the `number index' $j$, and multiplication over the `position index' $i$ gives
 the probability that a safe from the first ensemble 
has the same opening sequence, as  the safe from the second ensemble. 
This motivates the following definition of the scalar product.
 
 \begin{definition}
The scalar product of two row Markov matrices $q,p$ is the product of the diagonal elements of the matrix $qp^T$:
\begin{eqnarray}\label{17}
(q,p)=\prod _{i=0}^{d-1}\left[\sum_{j=0}^{d-1}q(i,j)p(i,j)\right];\;\;\;q,p\in{\cal M}.
\end{eqnarray}
\end{definition}
We note that in many cases a scalar product is a bilinear function, but this is not the case here. Also
\begin{eqnarray}
\prod _{i=0}^{d-1}\left[\sum_{j=0}^{d-1}q(i,j)p(i,j)\right]\ne \sum _{i=0}^{d-1}\left[\prod_{j=0}^{d-1}q(i,j)p(i,j)\right].
\end{eqnarray}
The first expression has $d^d$ products, while the second one has only $d$ products.

It is easily seen that 
\begin{eqnarray}
(q,p)=(p,q);\;\;\;(q, {\mathfrak U})=\frac{1}{d^d};\;\;\;(q,{\bf 1})=\prod _{i=0}^{d-1}q(i,i).
\end{eqnarray}
Also
\begin{eqnarray}\label{46}
(q,q)=\prod _{i=0}^{d-1}\left[\sum_{j=0}^{d-1}q(i,j)^2\right].
\end{eqnarray}
If we take two safes from an ensemble described by the row Markov matrix $q$, then
$(q,q)$ is the probability that they will have the same opening sequence.
If $q(i,j)=\delta(j,n_i)$ (in which case the sequence that opens all safes  is $(n_0,...,n_{d-1})$), then $(q,q)=1$.
$(q,q)$ takes small values when the distributions $q(i,j)$ are close to the uniform distribution, for all $i$.
In particular $({\mathfrak U}, {\mathfrak U})=\frac{1}{d^d}$.

We can write the result in lemma \ref{LL}  as
\begin{eqnarray}\label{RRR}
(M_f,M_g)=\delta(f,g);\;\;\;f,g\in {\cal F}.
\end{eqnarray} 

\begin{lemma}
For row Markov matrices the following relations hold:
\begin{itemize}
\item[(1)]
\begin{eqnarray}
0\le (q,p)\le 1.
\end{eqnarray}
\item[(2)]
\begin{eqnarray}\label{21}
(\lambda q_1+(1-\lambda)q_2, p)\ge\lambda ^d(q_1, p)+(1-\lambda )^d(q_2, p);\;\;\;0\le \lambda \le 1.
\end{eqnarray}

\end{itemize}
\end{lemma}
\begin{proof}
\begin{itemize}
\item[(1)]
We get
\begin{eqnarray}
\prod _{i=0}^{d-1}\left [\sum_{j_i=0}^{d-1}q(i,j)p(i,j)\right ]\le \prod _{i=0}^{d-1}\left [\frac{1}{2}\sum_{j_i=0}^{d-1}\left (q(i,j)^2+p(i,j)^2\right )\right ]\le 1
\end{eqnarray}
We used here the fact that
\begin{eqnarray}
\sum_{j=0}^{d-1}q(i,j)=1\;\;\rightarrow\;\;\sum_{j=0}^{d-1}q(i,j)^2\le 1,
\end{eqnarray}
and similarly for $p(i,j)$.
\item[(2)]
We get
\begin{eqnarray}
[\lambda q_1+(1-\lambda)q_2] p^T=\lambda (q_1p^T)+(1-\lambda )(q_2p^T).
\end{eqnarray}
When we get the product of the diagonal elements of the matrix in the left hand side, we get
the products of the diagonal elements of the two matrices in the right hand side, plus non-negative `cross terms' that contain the factor $\lambda ^n(1-\lambda)^{d-n}$
times $n$ $q_1$-variables times $d-n$ $q_2$-variables.
This proves the inequality.

\end{itemize}
\end{proof}
\begin{lemma}\label{L25}
If $q$ is a $d\times d$ matrix and $f\in {\cal F}$, then
\begin{eqnarray}\label{gh1}
\sum _f\prod _{i=0}^{d-1} q(i,f(i))=\prod _{i=0}^{d-1} \left [\sum _{j=0}^{d-1} q(i,j)\right ].
\end{eqnarray}
\end{lemma}
\begin{proof}
We first point out that there is the same number of $d^d$ terms in each side of this equality.
For each term $q(0,f(0))...q(d-1,f(d-1))$ on the left hand side, we find the same term on the right hand side written as $q(0,j_0)...q(d-1,j_{d-1})$
with $j_0=f(0),...,j_{d-1}=f(d-1)$.
The fact that $f$ is a function (as opposed to a general binary relation) implies that there is only one term in the right hand side equal to $q(0,f(0))...q(d-1,f(d-1))$.
\end{proof}

If $q$ is a row Markov matrix and $f\in {\cal F}$, we define the product of probabilities 
\begin{eqnarray}\label{gh1}
{\mathfrak M}_q(f)=(q,M_f)=\prod _{i=0}^{d-1} q(i,f(i)).
\end{eqnarray}
In the absence of correlations, ${\mathfrak M}_q(f)$ is the joint probability that a random safe will open with the sequence $(f(0),...,f(d-1))$.
Indeed using lemma \ref{L25} and the fact that $q$ is here a row Markov matrix we prove that
\begin{eqnarray}\label{pro}
\sum _{f\in {\cal F}_R}{\mathfrak M}_q(f)=1.
\end{eqnarray}

The following proposition discusses an expansion of any row Markov matrix $q$ in terms of the matrices $M_f$ with the ${\mathfrak M}_q(f)$ as coefficients.
\begin{proposition}
\begin{itemize}
\item[(1)]
A row Markov matrix $q$ can be expanded in terms of the permutation with repetition matrices $M_f$, as
\begin{eqnarray}\label{gh}
q=\sum _{f\in {\cal F}}{\mathfrak M}_q(f)M_f=\sum _{f\in {\cal F}}(q,M_f)M_f.
\end{eqnarray}
\item[(2)]
For the matrix  ${\mathfrak U}$ all the ${\mathfrak M}_{\mathfrak U}(f)$ are equal to each other, and the expansion in Eq.(\ref{gh}) gives
\begin{eqnarray}
{\mathfrak U}=\frac{1}{d^d}\sum _{f \in {\cal F}_R}M_f.
\end{eqnarray}
\item[(3)]
For  $q=M_g$ we get ${\mathfrak M}_q(f)=\delta (f,g)$.
In this case the expansion in Eq.(\ref{gh}) has only one term.
\end{itemize}
\end{proposition}
\begin{proof}
\begin{itemize}
\item[(1)]
The proof is similar to the proof of lemma \ref{L25}.
Using  ${\mathfrak M}_q(f)=q(0,f(0))...q(d-1,f(d-1))$ we calculate the $(i,j)$ element of the matrix in the right hand side of Eq.(\ref{gh}):
\begin{eqnarray}
&&\sum _{f\in {\cal F}} \left [q(0,f(0))...q(d-1,f(d-1))\right ]\delta(f(i),j)\nonumber\\&&=
\left (\sum _k q(0,k)\right )...\left (\sum _k q(i-1,k)\right )q(i,j)\left (\sum _k q(i+1,k)\right )...\left (\sum _k q(d-1,k)\right )=q(i,j).
\end{eqnarray}
Here we have the sum over all functions such that $f(i)=j$. So for all $k\ne i$ the $f(k)$ takes all possible values.
This proves Eq.(\ref{gh}).

\item[(2)]
In the special case that $q={\mathfrak U}$ we easily see that
\begin{eqnarray}
{\mathfrak M}_{\mathfrak U}({f})=\frac{1}{d^d}.
\end{eqnarray}
\item[(3)]
The proof of this is straightforward.
\end{itemize}
\end{proof}

The expansion in Eq.(\ref{gh}) says that in the ensemble of safes with probabilities described by the matrix $q$,
the probability that the opening sequence is $(f(0),...,f(d-1))$
(i.e., the percentage of safes with opening sequence $(f(0),...,f(d-1))$) is ${\mathfrak M}_q(f)$.
Special cases are:
\begin{itemize}
\item
If $q=M_f$, then all safes have the same opening sequence $(f(0),...,f(d-1))$.
\item
If $q={\mathfrak U}$, all opening sequences are equally probable (with probability $\frac{1}{d^d}$).
\end{itemize}

\begin{cor}
For row Markov matrices:
\begin{itemize}
\item[(1)]
\begin{eqnarray}\label{bb}
{\mathfrak M}_{\lambda q_1+(1-\lambda)q_2}(f)\ge \lambda ^d{{\mathfrak M}_{q_1}}(f)+(1-\lambda)^d{{\mathfrak M}_{q_2}}(f)
\end{eqnarray}
\item[(2)]
If 
\begin{eqnarray}
q=\sum _{f\in {\cal F}}{\mathfrak M}_q(f)M_f;\;\;\;p=\sum _{f\in {\cal F}}{\mathfrak M}_p(f)M_f
\end{eqnarray}
then
\begin{eqnarray}\label{579}
(q, p)=\sum _{f\in {\cal F}}{\mathfrak M}_q(f){\mathfrak M}_p(f).
\end{eqnarray}

\end{itemize}
\end{cor}
\begin{proof}
\begin{itemize}
\item[(1)]
We insert $q=M_f$ in Eq.(\ref{21}), and use Eq.(\ref{gh1}).
\item[(2)]
Using lemma \ref{L25} we get:
\begin{eqnarray}
(q,p)=\prod _{i=0}^{d-1}\left[\sum_{j=0}^{d-1}q(i,j)p(i,j)\right]=\sum _f\prod _i q(i,f(i))p(i,f(i))=\sum _{f}{\mathfrak M}_q(f){\mathfrak M}_p(f)
\end{eqnarray}
\end{itemize}

\end{proof}
The $\lambda q_1+(1-\lambda)q_2$ with $0\le \lambda \le 1$,
describes the merger of an ensemble described by $q_1$ which has $\lambda N$ safes, with 
an ensemble described by $q_2$ which has $(1-\lambda) N$ safes.
In this case the inequality in Eq.(\ref{bb}) holds.

\subsection{Expansions for row Markov matrices in the presence of correlations}

We consider an ensemble of many safes with correlations between the integers in the various positions of the sequences that open the safes, i.e., lack of independence.
Let $q(f)=q[f(0),...,f(d-1)]$ be $d^d$ joint probabilities (labelled with the functions $f\in {\cal F}$). 
In terms of a random safe, they are the joint probabilities that the number in the $0$-position of the opening sequence is $f(0)$,
the number in the $1$-position is $f(1)$, etc.
We call them Markov tensors and 
\begin{eqnarray}\label{42}
\sum _{f\in {\cal F}}q(f)=1.
\end{eqnarray}
The
\begin{eqnarray}\label{ghd}
q=\sum _{f\in {\cal F}}q(f)M_f,
\end{eqnarray}
is a row Markov matrix. Indeed
\begin{eqnarray}
\sum _jq(i,j)=\sum _j\sum _{f\in {\cal F}}q(f)M_f(i,j)=\sum _{f\in {\cal F}}q(f)=1.
\end{eqnarray}

Eq.(\ref{ghd}) provides an expansion of $q$ in terms of $M_f$ with the joint probabilities $q(f)$ as coefficients.
Eq.(\ref{gh}) provides an alternative  expansion of the same row Markov matrix $q$ in terms of $M_f$ with the products of probabilities ${\mathfrak M}_q(f)$ as coefficients,
and in general $q(f)\ne {\mathfrak M}_q(f)$.
This is because in the presence of correlations joint probabilities are not equal to the product of probabilities.
So the non-uniqueness of the expansion is related to various types of correlations.
Two ensembles might have different Markov tensors $q(f)$ (i.e. different correlations) but the same Markov matrix $q(i,j)$, and then we get two different expansions of the same Markov matrix.

The $d^d$ quantities
\begin{eqnarray}\label{COR}
{\cal C}_q(f)=q(f)-{\mathfrak M}_q(f)=q(f)-\prod _i q(i,f(i));\;\;\;-1\le {\cal C}_q(f)\le 1;\;\;\;\sum _{f\in {\cal F}_R}{\cal C}_q(f)=0,
\end{eqnarray}
is one possible way of quantifying correlations, and we call them correlation coefficients.
In the case of independence (absence of correlations) we get  ${\cal C}(f)=0$.

\begin{example}\label{ex23}
We consider the following $3\times 3$ row Markov matrix 
\begin{eqnarray}\label{AB1}
q=\begin{pmatrix}
a&1-a&0\\
0&a&1-a\\
0&1-b&b\\
\end{pmatrix}.
\end{eqnarray}
Assuming absence of correlations, we present in table \ref{t1} the probabilities ${\mathfrak M}_q(f)$ and the corresponding permutation with repetition matrices $M_f$ for the expansion in Eq.(\ref{gh}).
In general there are $d^d=27$ terms, but in this example only $8$ of them have non-zero probability.

In the interpretation in terms of safes, the opening sequence consists of $3$ integers which take one of the values $0,1,2$.
The probabilities that the first integer is $0,1,2$ are $a, 1-a, 0$ correspondingly, etc.
We assume independence between the integers in the various positions of the opening sequence.
Then the joint probability that the opening sequence is $(0,1,1)$ is $a^2(1-b)$,
the joint probability that the opening sequence is $(0,1,2)$ is $a^2b$, etc.

Using Eq.(\ref{46}) we also calculated the 
\begin{eqnarray}
(q,q)=(2a^2-2a+1)^2(2b^2-2b+1).
\end{eqnarray}
This is the probability that two safes will have the same opening sequence of integers.
\end{example}

\begin{example}\label{ex24}
We consider again the row Markov matrix in Eq.(\ref{AB1}) and for later use we assume that
\begin{eqnarray}\label{AB11}
0\le 2a\le  b\le \frac{1}{2}.
\end{eqnarray}
A different expansion for this Markov matrix is
\begin{eqnarray}\label{kk}
q=a\begin{pmatrix}
1&0&0\\
0&1&0\\
0&0&1\\
\end{pmatrix}+(b-a)
\begin{pmatrix}
0&1&0\\
0&0&1\\
0&0&1\\
\end{pmatrix}+(1-b)
\begin{pmatrix}
0&1&0\\
0&0&1\\
0&1&0\\
\end{pmatrix}.
\end{eqnarray}
The corresponding joint probabilities $q(f)$ and correlation coefficients ${\cal C}_q(f)=q(f)-{\mathfrak M}_q(f)$ are shown in table \ref{t1}.
Here we have correlations between the integers in the various positions of the opening sequence, and the joint probability is not equal to the product of probabilities.
In an ensemble of these safes, the percentage of safes with opening sequence $(0,1,2)$ is $a$,
the percentage of safes with opening sequence $(1,2,2)$ is $b-a$, and the percentage of safes with opening sequence $(1,2,1)$ is $1-b$.

\end{example}

\section{Lorenz values of probability vectors}

Lorenz values and the Gini index\cite{Gini,Gini1,Gini2,Gini3} are statistical quantities which have been used extensively in Mathematical Economics for the study of inequality in wealth distribution.
In our context they are used to quantify the sparsity (certainty) in a probability vector.

\subsection{Lorenz values of probability vectors}\label{majo}

 Let $\pi$ be the permutation that orders the probabilities of a probability vector ${\bf x}$, in ascending order: 
\begin{eqnarray}\label{pp}
x[\pi(0)]\le x[\pi(1)]\le ...\le x[\pi(d-1)].
\end{eqnarray}
We sometimes use the notation $\pi_{\bf x}$ for this permutation.

The Lorenz values of this probability vector are defined as
\begin{eqnarray}
{\cal L}(\ell;{\bf x})=x[\pi(0)]+...+ x[\pi(\ell)];\;\;\;{\cal L}(d-1;{\bf x})=1
\end{eqnarray}
where $\ell=0,...,d-1$.
The Lorenz values are cumulative probabilities with respect to the order $\pi$ in Eq.(\ref{pp}).

The Lorenz values for the `most uncertain probability vector' ${\bf u}$ and the `certain probability vector' ${\bf c}$, are
\begin{eqnarray}
{\cal L}(\ell;{\bf u})=\frac{\ell+1}{d};\;\;\;{\cal L}(\ell;{\bf c})=\delta(\ell, d-1),
\end{eqnarray}
where $\delta$ is the Kronecker delta.
It is easily seen that for any permutation matrix $M_{\pi}$
\begin{eqnarray}\label{hh1}
{\cal L}(\ell;{\bf x}M_\pi)={\cal L}(\ell;{\bf x})
\end{eqnarray}
But for a general permutation with repetition matrix $M_f$ the ${\cal L}(\ell;{\bf x}M_f)$ might be different from ${\cal L}(\ell;{\bf x})$.
\begin{proposition}\label{L1}
\begin{eqnarray}\label{24}
0\le {\cal L}(\ell;{\bf x})\le \frac{\ell+1}{d}.
\end{eqnarray}
\end{proposition}
\begin{proof}
If $x[\pi(\ell)]\le \frac{1}{d}$, then since $x[\pi(k)]\le x[\pi(\ell)]$ for $k<\ell$ we prove easily Eq.(\ref{24}) for this case.

If $x[\pi(\ell)]> \frac{1}{d}$, then since $x[\pi(k)]\ge  x[\pi(\ell)]$ for $k>\ell$, we get
\begin{eqnarray}
1=x[\pi(0)]+...+x[\pi(\ell)]+...+x[\pi(d-1)]\ge {\cal L}(\ell;{\bf x})+(d-1-\ell)x[\pi(\ell)]\ge {\cal L}(\ell;{\bf x})+\frac{d-1-\ell}{d}
\end{eqnarray}
From this follows Eq.(\ref{24}) for this case.
\end{proof}

\begin{definition}
Two probability vectors ${\bf x}, {\bf y}$ are comonotonic if $\pi_{\bf x}=\pi_{\bf y}$.
In this case the probability vector $\lambda {\bf x}+(1-\lambda) {\bf y}$ where $0\le \lambda \le 1$, also has the same permutation of ordering of its elements.
\end{definition}
\begin{proposition}\label{pro78}
\mbox{}
\begin{itemize}
\item[(1)]
If ${\bf x}, {\bf y}$ are probability vectors, then for all $\ell$
\begin{eqnarray}\label{hh2}
{\cal L}[\ell;\lambda {\bf x}+(1-\lambda){\bf y}]\ge {\cal L}(\ell;{\bf x})+(1-\lambda){\cal L}(\ell;{\bf y});\;\;\;0\le \lambda\le 1.
\end{eqnarray}
For comonotonic vectors ${\bf x}, {\bf y}$, this becomes equality.
\item[(2)]
If ${\mathfrak D}$ is a doubly stochastic matrix, then ${\cal L}[\ell;{\bf x}{\mathfrak D}]\ge{\cal L}[\ell;{\bf x}]$, for all $\ell$.
This is not true for row Markov matrices, in general.
\end{itemize}
\end{proposition}
\begin{proof}
\mbox{}
\begin{itemize}
\item[(1)]
We consider the probability vector
\begin{eqnarray}
{\bf z}=\lambda {\bf x}+(1-\lambda){\bf y}.
\end{eqnarray}
Then
\begin{eqnarray}\label{V1}
{\cal L}(\ell;{\bf z})&=&z[\pi_z(0)]+...+ z[\pi_z(\ell)]\nonumber\\&=&\{\lambda x[\pi_z(0)]+(1-\lambda) y[\pi_z(0)]\}+...+
\{\lambda x[\pi_z(d-1)]+(1-\lambda) y[\pi_z(d-1)]\}
\end{eqnarray}
$\pi_z$ is the ordering permutation for ${\bf z}$, and in general it is different from the 
ordering permutation $\pi_x$ for ${\bf x}$ and $\pi_y$ for ${\bf y}$. Therefore
\begin{eqnarray}\label{V2}
&&x[\pi_z(0)]+ ...+x[\pi_z(\ell)]\ge {\cal L}(\ell;{\bf x})\nonumber\\
&&y[\pi_z(0)]+ ...+y[\pi_z(\ell)]\ge {\cal L}(\ell;{\bf y})
\end{eqnarray}
Combining Eqs(\ref{V1}), (\ref{V2}) we prove Eq.(\ref{hh2}).

For comonotonic probability vectors ${\bf x}, {\bf y}$, the ${\bf z}$ has the same ordering permutation as the ${\bf x}, {\bf y}$.
Then Eqs.(\ref{V2}) become equalities and  therefore Eq.(\ref{hh2}) becomes equality.

\item[(2)]
The proof for doubly stochastic matrices is based on Eq.(\ref{hh2}), taking into account Eqs.(\ref{hh}), (\ref{hh1}).
We note that Eq.(\ref{hh1}) is not valid for general permutation with repetition matrices $M_f$ that appear in the expansions of row Markov matrices.
The following example shows that  ${\cal L}(\ell;{\bf x})$ can be  greater than ${\cal L}(\ell;{\bf x}q)$.

For the row Markov matrix in Eq.(\ref{AB1}) with $a=b=\frac{1}{2}$, we find that
\begin{eqnarray}\label{AB3}
{\bf u}q=\left (\frac{1}{6}, \frac{1}{2}, \frac{1}{3}\right ).
\end{eqnarray}
Therefore
\begin{eqnarray}
{\cal L}(0;{\bf u}q)=\frac{1}{6};\;\;\;{\cal L}(1;{\bf u}q)=\frac{1}{2};\;\;\;{\cal L}(2;{\bf u}q)=1.
\end{eqnarray}
In this example ${\cal L}(\ell;{\bf u})=\frac{\ell+1}{3}$ is greater than ${\cal L}(\ell;{\bf u}q)$.
\end{itemize}
\end{proof}

\subsection{Majorization of probability vectors: `more sparse'}
Majorization is a preorder that has been used in various areas (e.g., \cite{MAJ}), including quantum physics\cite{MA0,MA1,MA2,MA3}.
It provides an ordinal description of the sparsity in probability vectors (i.e., that one probability vector is more sparse than another).

\begin{definition}
 If ${\bf x}, {\bf y}$ are two probability vectors, then ${\bf x}\succ {\bf y}$ (${\bf x}$ majorizes ${\bf y}$ or ${\bf x}$ is `more sparse' than ${\bf y}$ ) if ${\cal L}(\ell;{\bf x})\le {\cal L}(\ell;{\bf y})$ for all $\ell$.
\end{definition}
In this case the large probabilities are larger in ${\bf x}$ than in ${\bf y}$, and therefore the `certainty' related to  ${\bf x}$ is larger than the certainty related to ${\bf y}$ . 
In other words, the probability   distribution in ${\bf x}$ is `more certain' than the distribution in ${\bf y}$. 
Clearly for any vector ${\bf x}$
\begin{eqnarray}\label{19}
{\bf u}\prec{\bf x}\prec {\bf c}
\end{eqnarray}
where ${\bf u}$ and ${\bf c}$ have been defined in Eq.(\ref{3D}).

It is easily seen that
\begin{eqnarray}\label{mn}
&&{\bf x}\prec{\bf x}\nonumber\\
&&{\bf z}\prec{\bf y}{\;\;\rm and\;\;}{\bf y}\prec{\bf x}\;\;\rightarrow\;\;{\bf z}\prec {\bf x}
\end{eqnarray}
and therefore ${\bf x}\succ {\bf y}$ is a preorder in ${\cal P}$.

Majorization is intimately related to doubly stochastic matrices. ${\bf x}\succ {\bf y}$ if and only if there exists a doubly stochastic matrix ${\mathfrak D}$ such that ${\bf y}={\mathfrak D}{\bf x}$.
Indeed if ${\bf y}={\mathfrak D}{\bf x}$ then using Eqs(\ref{hh}),(\ref{hh1}),(\ref{hh2}) we prove that
${\cal L}(\ell,{\bf x})\ge {\cal L}(\ell,{\bf y})$ for all $\ell$, i.e., that ${\bf x}\succ {\bf y}$.
The converse is also true\cite{MAJ}.

\begin{remark}\label{rem2}
Majorization is an ordinal approach to the concept `more certain' or `more sparse'. 
Entropic quantities can be used  to quantify these concepts.
There is a relationship between majorization and entropic quantities which involves the concept of
Schur concave functions, i.e., functions $\phi({\bf x})$ for which
\begin{eqnarray}
{\bf x}\succ{\bf y}\;\;\rightarrow\;\;\phi({\bf x})\le \phi ( {\bf y}).
\end{eqnarray}
Many entropic quantities are known to be  Schur concave functions(e.g., \cite{BZ}), and for them the entropy corresponding 
${\bf x}$ is less than the entropy corresponding to ${\bf y}$. We do not purse further the `entropic direction' in this paper.

\end{remark}

\section{The Gini index as an indicator of the sparsity of probability vectors}

The Gini index quantifies the variability in a probability distribution.
It is an alternative to variance and standard deviation, and below (in remark \ref{rem1}) we compare and contrast these two quantities.
It has been used in Mathematical Economics for the characterisation of wealth inequality.

In this section Lorenz values are used to define the Gini index which quantifies the sparsity in probability vectors.
It has attractive properties (e.g. proposition \ref{pro56} ) which can be used to quantify the concepts `uncertainty increase' or `information loss'. 
In this sense it is a complementary quantity to entropy.

\begin{proposition}\label{proG1}
The Gini index is defined by the following relations which are equivalent to each other:
\begin{itemize}
\item[(1)]
\begin{eqnarray}\label{85}
{\cal G}({\bf x})&=&1-\frac{2}{d+1}\sum _{\ell=0}^{d-1}{\cal L}(\ell;{\bf x})=\frac{d-1}{d+1}-\frac{2}{d+1}\sum _{\ell=0}^{d-2}{\cal L}(\ell;{\bf x})\nonumber\\&=&
1-\frac{2}{d+1}\{dx[\pi(0)]+(d-1)x[\pi(1)]+...+x[\pi(d-1)]\}
\end{eqnarray}

\item[(2)]
\begin{eqnarray}\label{3d}
{\cal G}({\bf x})=\frac{1}{2(d+1)}\sum _{r,s}|x(r)-x(s)|
\end{eqnarray}
\end{itemize}
\end{proposition}

\begin{proof}
Using the ordering of the probabilities in Eq.(\ref{pp}) we get
\begin{eqnarray}
\sum _{r,s}|x(r)-x(s)|&=&\sum _{r,s}|x[\pi(r)]-x[\pi(s)]|\nonumber\\&=&
2\sum _{r=1}^{d-1}\sum _{s< r}\{x[\pi(r)]-x[\pi(s)]\}\nonumber\\&=&
2\sum _{r=1}^{d-1} (r-1)x[\pi(r)]-2\sum _{r=1}^{d-1}{\cal L}(r-1;{\bf x})
\end{eqnarray}
We note that
\begin{eqnarray}
\sum _{r=1}^{d-1} (r-1)x[\pi(r)]+\sum _{r=1}^{d-1}{\cal L}(r-1;{\bf x})=d-1
\end{eqnarray}
and we get
\begin{eqnarray}
\sum _{r,s}|x(r)-x(s)|=4\sum _{r=1}^{d-1} (r-1)x[\pi(r)]-2(d-1).
\end{eqnarray}
From this follows Eq.(\ref{85}).
\end{proof}

From Eq.(\ref{24}) follows that
\begin{eqnarray}\label{www}
0\le {\cal G}({\bf x})\le \frac{d-1}{d+1}.
\end{eqnarray}
Large values of the Gini index indicate sparse (certain) probability vectors.
Small values of the Gini index indicate uncertain probability vectors.

\begin{example}
The Gini indices of the vectors ${\bf u}$ and ${\bf c}$ are
\begin{eqnarray}
{\cal G}({\bf u})=0;\;\;\;{\cal G}({\bf c})=\frac{d-1}{d+1}
\end{eqnarray}
\end{example}

\begin{proposition}\label{pro56}
\mbox{}
\begin{itemize}
\item[(1)]
If ${\bf x}$, ${\bf y}$ are probability vectors, then
\begin{eqnarray}\label{pp5}
{\cal G}[\lambda {\bf x}+(1-\lambda){\bf y}]\le \lambda {\cal G}({\bf x})+(1-\lambda){\cal G}({\bf y});\;\;\;0\le \lambda\le 1.
\end{eqnarray}
For comonotonic vectors ${\bf x}, {\bf y}$, this becomes equality.
\item[(2)]
If ${\mathfrak D}$ is a doubly stochastic matrix, then 
\begin{eqnarray}\label{P1}
{\cal G}({\bf x}{\mathfrak D})\le {\cal G}({\bf x})
\end{eqnarray}
If the doubly stochastic matrix is a permutation matrix $M_\pi$ then
\begin{eqnarray}\label{P2}
{\cal G}({\bf x}M_{\pi})={\cal G}({\bf x})
\end{eqnarray}
But for a row Markov matrix $q$ the ${\cal G}({\bf x}q)$ might be greater than ${\cal G}({\bf x})$.

\end{itemize}
\end{proposition}
\begin{proof}
\begin{itemize}
\item[(1)]
This follows from Eq.(\ref{85}) and the first part of proposition \ref{pro78}.
\item[(2)]
Eq.(\ref{P1}) follows from Eq.(\ref{85}) and the second part of proposition \ref{pro78}.
Eq.(\ref{P2}) follows from Eq.(\ref{hh1}).

An example where the ${\cal G}({\bf x}q)$ is greater than ${\cal G}({\bf x})$, is as follows.
For the uncertain probability vector in Eq.(\ref{3D}), we get ${\cal G}({\bf u})=0$.
For the probability vector ${\bf r}={\bf u}q$ in Eq.(\ref{AB3}), we get ${\cal G}({\bf r})=\frac{1}{6}$.

\end{itemize}
\end{proof}
The change in the Gini index of a probability vector ${\bf x}$ when it is multiplied by a row Markov matrix $q$, is:
\begin{eqnarray}\label{R}
\Delta {\cal G}({\bf x},{\bf x}q)={\cal G}({\bf x})-{\cal G}({\bf x}q).
\end{eqnarray}
The  $\Delta {\cal G}({\bf x},{\bf x}q)$ quantifies the change in the uncertainty as we go from the probability vector ${\bf x}$ to the probability vector ${\bf x}q$.
If $q$ is a doubly stochastic matrix ${\mathfrak D}$ then $\Delta {\cal G}({\bf x},{\bf x}{\mathfrak D})\ge 0$.
Also
\begin{eqnarray}\label{R}
\Delta {\cal G}({\bf x},{\bf x}{\mathfrak U})={\cal G}({\bf x}).
\end{eqnarray}
If $q$ is a row Markov matrix which is not doubly stochastic, then the probability vector ${\bf x}q$ might be more sparse (more certain) than ${\bf x}$, and
the $\Delta {\cal G}({\bf x},{\bf x}q)$ might be negative.

\begin{remark}\label{rem1}
Eq.(\ref{3d}) for the Gini index should be compared and contrasted with the variance $V$ related to these probabilities which is given by
\begin{eqnarray}
V=\frac{1}{2d^2}\sum _{r,s}[x(r)-x(s)]^2.
\end{eqnarray}
It is seen that the variance is the sum of the squares of the differences, whilst the Gini index is the sum of the absolute values of the differences.
It is easily seen that
\begin{eqnarray}
d^2V\le (d+1){\cal G}({\bf x}).
\end{eqnarray}

In our context properties like proposition \ref{pro56}, make the Gini index an attractive quantity for the characterisation of uncertainty.
\end{remark}

\subsection{Bounds for the averages in terms of the Gini index}\label{sec12}

We consider a quantity that takes the values $0,...,d-1$ with probabilities $x(0),...,x(d-1)$, correspondingly.
The average value of this quantity is
\begin{eqnarray}\label{order}
\langle{\bf x}\rangle=0\cdot x(0)+x(1)+...+(d-1)x(d-1).
\end{eqnarray}
In many cases with a finite number of items there is no `natural order'. The values $0,...,d-1$ are just labels and with
a permutation $\wp$ we can change  the value $\ell$ into $\wp (\ell)$. Then the average value changes into
\begin{eqnarray}\label{ord}
\langle{\bf x}\rangle_\wp &=&0\cdot x[\wp (0)]+x[\wp (1)]+...+(d-1) x[\wp (d-1)]\nonumber\\
&=&\wp^{-1}(0)x(0)+\wp^{-1}(d-2)x(1)+...+\wp^{-1}(d-1)x(d-1).
\end{eqnarray}
We note that the order $\pi$ of the values of the probabilities in Eq.(\ref{pp}) defines uniquely the Gini index.
In contrast, the average depends on another order $\wp$ that is defined by the physical problem.
The following proposition provides an interval where $\langle{\bf x}\rangle_\wp$  belongs.
\begin{proposition}\label{GG3}
For any permutation $\wp$, the average $\langle{\bf x}\rangle_\wp$  belongs to the interval
\begin{eqnarray}\label{39}
 \frac{d-1}{2}-\frac{d+1}{2}{\cal G}({\bf x})\le  \langle{\bf x}\rangle _\wp \le \frac{d-1}{2}+\frac{d+1}{2}{\cal G}({\bf x})
 \end{eqnarray}
\end{proposition}
\begin{proof}
We first consider the sum 
\begin{eqnarray}
S(\ell)=x[\wp (\ell)]+...+x[\wp (d-1)];\;\;\;S(0)=1.
\end{eqnarray}
Clearly
\begin{eqnarray}
{\cal L}(d-\ell-1;{\bf x})\le S(\ell)\le 1-{\cal L}(\ell -1;{\bf x})
\end{eqnarray}
because ${\cal L}(d-\ell -1;{\bf x})$ is the sum of the lowest $d-\ell$ probabilities, and $1-{\cal L}(\ell -1;{\bf x})$ is the sum of the highest $d-\ell $ probabilities. Therefore
\begin{eqnarray}\label{A1}
\langle{\bf x}\rangle_{\wp}=\sum _{r=1}^{d-1}S(r)\ge\sum _{r=1}^{d-1}{\cal L}(d-r-1;{\bf x})=A.
\end{eqnarray}
Also
\begin{eqnarray}\label{A2}
\langle{\bf x}\rangle_{\wp}=\sum _{r=1}^{d-1}S(r)\le\sum _{r=1}^{d-1}[1-{\cal L}(r-1;{\bf x})]=d-1-A.
\end{eqnarray}
Eq.(\ref{85}) shows that
\begin{eqnarray}
A=\frac{d+1}{2}\left [\frac{d-1}{d+1}-{\cal G}({\bf x})\right]
\end{eqnarray}
Using this with Eqs(\ref{A1}),(\ref{A2}) we prove the proposition.
\end{proof}

\subsection{Local and total Gini indices for random safes}
We consider the $d\times d$ row Markov matrix $q$, describing an ensemble of safes (random safe).
As explained earlier $q(i,j)$ is the probability that the integer in the $i$-position of the sequence that opens the safe, is $j$.
For a fixed row $i$, the $q(i,j)$ is a probability vector and we can calculate its Gini index. 
We order the $i$-row of $q$ in ascending order as
\begin{eqnarray}
q[i,\pi_i(0)]\le...\le q[i,\pi_i(d-1)],
\end{eqnarray}
where $\pi _i$ is a permutation of the $d$ probabilities $q(i,j)$ in the $i$-row.
Then
\begin{eqnarray}\label{G1}
{\cal G}_i={\cal G}[q(i,j)]=
1-\frac{2}{d+1}\{dq[i,\pi_i(0)]+...+q[i,\pi_i(d-1)]\};\;\;\;0\le {\cal G}_i\le \frac{d-1}{d+1}.
\end{eqnarray}
This is a local Gini index describing the sparsity of the probability vector for the integer in the $i$-position of the sequence that opens the random safe.
 Then the `Gini vector' 
\begin{eqnarray}\label{G1V}
{\cal G}=({\cal G}_0,...,{\cal G}_{d-1}),
\end{eqnarray}
consists of all local Gini indices, and it does not depend on correlations between the various integers in the sequence that opens the random safe.

We also consider the joint probabilities $q(f_i)$ ($i=0,...,d^d-1$) which form a probability vector with $d^d$ components, and define its Gini index.
We order the $d^d$ joint probabilities $q(f_i)$ in ascending order as
\begin{eqnarray}
q[\pi(f_0)]\le...\le q[\pi(f_{d^d-1})]\},
\end{eqnarray}
where $\pi$ is a permutation of the $d^d$ functions $f_i\in {\cal F}$. Then
\begin{eqnarray}\label{G2}
 {\cal G}_T={\cal G}[q(f)]=
1-\frac{2}{d^d+1}\left \{d^dq[\pi(f_0)]+...+q[\pi(f_{d^d-1})]\right \};\;\;\;0\le {\cal G}_T\le\frac{d^d-1}{d^d+1}.
\end{eqnarray}
${\cal G}_T$ is a total Gini index describing the sparsity of the joint probabilities describing the ensemble of safes.
${\cal G}_T$ depends on correlations between the various integers in the sequence that opens the random safe.

\begin{example}
This is a continuation of example \ref{ex24}.
We consider the row Markov matrix in Eq.(\ref{AB1}) and its expansion in Eq.(\ref{kk}).
The Gini vector is easily found to be
\begin{eqnarray}
{\cal G}=\left (\frac{1-a}{2}, \frac{1-a}{2}, \frac{1-b}{2}\right ).
\end{eqnarray}
The joint probabilities are $a$, $b-a$ and $1-b$, and using the inequalities in Eq.(\ref{AB11}) we see that
\begin{eqnarray}
0\le a\le b-a\le 1-b.
\end{eqnarray} 
Therefore the total Gini index is 
\begin{eqnarray}
{\cal G}_T=1-\frac{2}{28}[3a+2(b-a)+(1-b)]=\frac{13-a-b}{14}.
\end{eqnarray}
\end{example}

\section{Quantum uncertainties in terms of the Gini index}

We  consider a quantum system (qudit) with variables in  ${\mathbb Z}_d$ (e.g.,\cite{V1}). 
$H_d$ is the $d$-dimensional Hilbert space describing this system. 
$|j\rangle$ where $j\in {\mathbb Z}_d$, is an orthonormal basis in $H_d$.
With a finite Fourier transform $F$ 
\begin{eqnarray}\label{FF}
&&F=\frac{1}{\sqrt{d}}\sum _{j,k}\omega_d(jk) \ket{j}\bra{k};\;\;\;\omega_d(\alpha)=\exp \left (i\frac{2\pi\alpha}{d}\right);\;\;\;\alpha,j,k\in{\mathbb Z}_d\nonumber\\
&&F^4={\bf 1};\;\;\;FF^{\dagger}={\bf 1},
\end{eqnarray}
we introduce the dual  basis
\begin{eqnarray}\label{FF}
&&\ket{j}_F=F\ket{j}.
\end{eqnarray}
The term dual refers to Fourier transform (in multipartite systems below we make the distinction between locally dual and globally dual for two different types of Fourier transform).
We call $\ket{j}$ positions and $\ket{j}_F$ momenta.
Using the relation
\begin{eqnarray}
\frac{1}{d}\sum _{k}\omega_d[(j+\ell)k]=\delta(j,-\ell),
\end{eqnarray}
we show that $F^2$ is the parity operator
\begin{eqnarray}\label{parity}
F^2=\frac{1}{d}\sum _{j,k,\ell}\omega_d[(j+\ell)k] \ket{j}\bra{\ell}=\sum _{j}\ket{j}\bra{-j}
\end{eqnarray}
For $d=2$ we get $F^2={\bf 1}$.

Let $\varpi(j)$ and ${\widetilde \varpi}(j)$ be the following orthogonal projectors and their `duals':
\begin{eqnarray}\label{333}
\varpi(j)=\ket{j}\bra{j};\;\;\;{\widetilde \varpi}(j)=F\varpi(j)F^\dagger=\ket{j}_F\; _F\bra{j}.
\end{eqnarray}
Measurements with these projectors on a system with density matrix $\rho$, will give the outcome `yes' with probabilities
\begin{eqnarray}\label{rr}
x_\rho (j)={\rm Tr}[\rho \varpi(j)];\;\;\;{\widetilde x}_\rho (j)={\rm Tr}[\rho {\widetilde \varpi}(j)]=x_{F^{\dagger}\rho F}(j),
\end{eqnarray}
correspondingly.
We refer to the ${\widetilde x}_\rho (j)$ as the dual probabilities. 
The dual probabilities for the density matrix $\rho$, are the probabilities of the Fourier transformed density matrix $F^\dagger\rho F$.
The Gini indices ${\cal G}(\rho)$ and ${\widetilde {\cal G}}(\rho)$ for these two probability vectors are calculated using Eq.(\ref{85}).
They quantify the sparsity of these two probability vectors and it is easily seen that  
\begin{eqnarray}
{\widetilde {\cal G}}(\rho)={\cal G}(F^{\dagger}\rho F);\;\;\;0\le {\cal G}(\rho), {\widetilde {\cal G}}(\rho)\le \frac{d-1}{d+1}.
\end{eqnarray}

In ref\cite{V2} we have used the Gini index to study the uncertainty principle for systems with finite dimensional Hilbert space.
We have shown that ${\cal G}(\rho)+{\widetilde {\cal G}}(\rho)$ cannot take values arbitrarily close to $2\frac{d-1}{d+1}$, i.e., 
we cannot have probability vectors $x_\rho (j)$, ${\widetilde x}_\rho (j)$ which are both `very sparse'.
Based on this we proved that the following `uncertainty coefficient' is greater than zero:
\begin{eqnarray}\label{Q1}
\eta _d=2\frac{d-1}{d+1}-\sup _{\rho }[{\cal G}(\rho)+{\widetilde {\cal G}}(\rho)]>0.
\end{eqnarray}
We use here the supremum of ${\cal G}(\rho)+{\widetilde {\cal G}}(\rho)$ over all density matrices.
$\eta _d$ does not depend on the density matrix.
The uncertainty relation is the fact that $\eta_d$ is non-zero,  and it can be expressed as
\begin{eqnarray}\label{Q2}
\Delta(\rho)\ge \eta _d>0;\;\;\;\Delta(\rho)=2\frac{d-1}{d+1}-[{\cal G}(\rho)+{\widetilde {\cal G}}(\rho)].
\end{eqnarray}
In some sense $\eta_d$ is the analogue of $\frac{1}{2}$ in the infinite systems.
We also gave an upper bound for $\eta_d$ but this is not relevant to the uncertainty principle, and it is not used here.

\section{Multipartite quantum systems as quantum permutations with repetitions and as quantum safes}\label{Q}

The interpretation of multipartite quantum systems as quantum permutations with repetitions and as quantum safes,
allows the use of the formalism in the previous sections, in the present context.
Furthermore, practical calculations of quantum quantities like partial traces, expectation values, etc,  tacitly use permutations with repetitions.

\subsection{Local Fourier transforms}

We consider a $d$-partite system comprised of $d$ components each of which is a qudit.
This system is described with the $d^d$-dimensional Hilbert space ${\mathfrak H}=H_d\otimes ...\otimes H_d$, and let $\ket{j_0,...,j_{d-1}}$ be an orthonormal basis in it.
Also let $F_L$ be local Fourier transforms applied on each of the $d$ components of the system:
\begin{eqnarray}
F_L=F\otimes ...\otimes F;\;\;\;F_L^4={\bf 1};\;\;\;F_LF_L^{\dagger}={\bf 1}.
\end{eqnarray}
The index $L$ in the notation stands for local.
Acting with $F_L$ on the basis $\ket{j_0,...,j_{d-1}}$ we get the locally dual basis
\begin{eqnarray}\label{103}
\ket{j_0,...,j_{d-1}}_{\rm L}=F_L\ket{j_0,...,j_{d-1}}=\ket{j_0}_F\otimes...\otimes\ket{j_{d-1}}_F=\frac{1}{\sqrt{d^d}}\bigotimes _{r=0}^{d-1}\left [\sum _{k_r=0}^{d-1}\omega_{d}(j_rk_r)\ket{k_r}\right ].
\end{eqnarray}
We introduce `quantum permutations with repetitions' or `quantum safes', by labelling  an orthonormal basis in $\mathfrak H$ with the $d^d$ functions in the set ${\cal F}$ (that describes permutations with repetitions):
\begin{eqnarray}\label{kkk}
\ket{f}=\ket{f(0),...,f(d-1)};\;\;\;\sum _f\ket{f}\bra{f}={\bf 1};\;\;\;f\in {\cal F};\;\;\;f(i)\in {\mathbb Z}_d.
\end{eqnarray}
The Hilbert space $H_d\otimes ...\otimes H_d$ describes superpositions of permutations with repetitions or 
superpositions of safes.
For example, let $\rho _A, \rho_B$ be the density matrices
\begin{eqnarray}\label{108}
&&\rho _A=|a|^2\ket{f}\bra{f}+|b|^2\ket{g}\bra{g};\;\;\;|a|^2+|b|^2=1;\;\;\;f,g\in {\cal F}\nonumber\\
&&\rho _B=\rho_A+ab^*\ket{f}\bra{g}+a^*b\ket{g}\bra{f}.
\end{eqnarray}
$\rho _A$ describes a probabilistic combination of the permutations with repetitions $f,g$ (random safes).
$\rho _B$ describes a superposition of the permutations with repetitions $f,g$ (quantum safes).
The difference between the separable density matrix $\rho_A$ and  the entangled state \cite{H} described by the density matrix $\rho _B$, is the off-diagonal terms.

We consider the commuting projectors and their locally dual counterparts
\begin{eqnarray}\label{333a}
&&\Pi(i,j)={\bf 1}\otimes...\otimes {\bf 1}\otimes \varpi(j)\otimes {\bf 1}\otimes...\otimes {\bf 1};\;\;\;\sum_j\Pi(i,j)={\bf 1};\;\;\;[\Pi(i,j),\Pi(k,\ell)]=0\nonumber\\
&&{\widetilde \Pi}(i,j)=F_L\Pi(i,j)F_L^{\dagger}={\bf 1}\otimes...\otimes {\bf 1}\otimes {\widetilde \varpi}(j)\otimes {\bf 1}\otimes...\otimes {\bf 1}
\end{eqnarray}
The index $i$ (with $i=0,...,d-1$) indicates the position of $\varpi(j)$.
They describe measurements with $\varpi(j)$ (and ${\widetilde {\varpi}} (j)$)on the $i$-component of the system.
$i$ is a `position  index' and below it appears mainly in products, while $j$ is a `basis index' (the analogue of `number index' in previous sections) and below it appears mainly in sums.

If $f\in {\cal F}$ is a permutation with repetitions, we consider the $d^d$ projectors and their duals
\begin{eqnarray}\label{333}
&&\Pi(f)=\Pi[f(0),...,f(d-1)]=\varpi[f(0)]\otimes...\otimes \varpi[f(d-1)]=\prod_i\Pi[i,f(i)];\;\;\;\sum_{f}\Pi(f)={\bf 1}\nonumber\\
&&{\widetilde \Pi}(f)=F_L\Pi(f)F_L^\dagger=\prod_i\widetilde \Pi[i,f(i)]
\end{eqnarray}
They describe uncorrelated or independent local measurements on the various components of the system.
The overall outcome is `yes', if the outcome of the measurement $\varpi[f(i)]$ on the $i$-component is `yes', for all components.

\subsection{The Markov matrices formalism in quantum context} \label{markov}

We interpret all quantities for row Markov  matrices introduced earlier, in the present quantum context.
Let $\rho $ be a density matrix and 
\begin{eqnarray}\label{red}
{\breve \rho}_i={\rm Tr}_{k\ne i}\rho,
\end{eqnarray}
be the reduced density matrix of the $i$-component of the system, which is found
by taking the partial trace with respect to all components except $i$ (denoted as ${\rm Tr}_{k\ne i}$).
Then:
\begin{itemize}

\item
The row Markov matrix $q_\rho$ has elements
\begin{eqnarray}\label{333}
q_\rho(i,j)={\rm Tr}[\rho\Pi(i,j)]={\rm Tr}[{\breve \rho}_i\varpi(j)];\;\;\;\sum _jq_\rho(i,j)=1;\;\;\;q_\rho \in {\cal M},
\end{eqnarray}
Usually the row Markov matrices are used to describe discrete time evolution of classical systems. 
In multipartite quantum systems considered here, the row Markov matrices are used in connection with the property $\sum _j{\rm Tr}[\rho\Pi(i,j)]=1$.

The locally dual row Markov matrix has elements
\begin{eqnarray}
{\widetilde q}_\rho(i,j)={\rm Tr}[\rho\widetilde \Pi(i,j)]={\rm Tr}[F_L^\dagger \rho F_L \Pi(i,j)].
\end{eqnarray}

\item
For $\rho=\ket{f} \bra{f}$ where $\ket{f}$ is the state in Eq.(\ref{kkk}), we get
\begin{eqnarray}
q_f(i,j)=\bra{f}\Pi(i,j)\ket{f}=\bra{f(i)}\varpi(j)\ket{f(i)}=\delta(f(i),j)=M_f(i,j);\;\;\;f\in{\cal F}.
\end{eqnarray}
In the quantum context, the matrix $M_f$ describes the probabilities for measurements with $\varpi(j)$ on the state $\ket{f(i)}$.

\item
If $\rho$ is a density matrix and $f\in {\cal F}$ is a permutation with repetitions, the Markov tensor
\begin{eqnarray}
q_\rho(f)={\rm Tr}[\rho \Pi(f)]={\rm Tr}\{\rho[\varpi[f(0)]\otimes...\otimes \varpi[f(d-1)]]\};\;\;\;\sum _{f\in {\cal F}}q_\rho(f)=1,
\end{eqnarray}
 is the joint probability that the measurement $\varpi [f(i)]$ on the $i$-component of the system will give `yes', for all components $i$.
 We note that the $q_\rho(f)$ cannot detect entangling off-diagonal elements $\ket{f}\bra {g}$ in the density matrix.
 For example, in Eq.(\ref{108}) they cannot distinguish the separable density matrix $\rho _A$, from the entangled $\rho _B$.
 
The Markov tensor $q_\rho(f)$ is related to the matrix $q_\rho$ as follows:
\begin{eqnarray}\label{RT1}
q_\rho=\sum _fq_\rho(f)M_f
\end{eqnarray}

We also introduce
the locally dual joint probabilities 
\begin{eqnarray}\label{RT2}
{\widetilde q}_\rho(f)=q_{F_L^\dagger \rho F_L}(f)={\rm Tr}[F_L^\dagger \rho F_L \Pi(f)].
\end{eqnarray}

\item
From the $d\times d$ matrix $q_\rho\in{\cal M}$, we get the product of probabilities
\begin{eqnarray}
{\mathfrak M}_\rho(f)=\prod _i q_\rho(i,f(i)),
\end{eqnarray}
and the correlation coefficients
\begin{eqnarray}\label{cor}
{\cal C}_\rho (f)=q_\rho(f)-{\mathfrak M}_\rho(f);\;\;\;-1\le {\cal C}_\rho (f)\le 1;\;\;\;\sum _{f\in {\cal F}}{\cal C}_\rho(f)=0,
\end{eqnarray}
In a similar way we introduce the locally dual product of probabilities ${\widetilde {\mathfrak M}}_\rho(f)$
and the locally dual correlation coefficients ${\widetilde {\cal C}}_\rho (f)$.

\item
 If  $\rho, \sigma$ are two density matrices, then the $(q_\rho, q_\sigma)$ is defined in analogous way to Eq.(\ref{17}) as
\begin{eqnarray}
(q_\rho,q_\sigma)=\prod _{i=0}^{d-1}\left[\sum_{j=0}^{d-1}{\rm Tr}[\rho\Pi(i,j)]{\rm Tr}[\sigma\Pi(i,j)]\right]=\prod _{i=0}^{d-1}\left[\sum_{j=0}^{d-1}{\rm Tr}[{\breve \rho}_i\varpi(j)]{\rm Tr}[{\breve \sigma}_i\varpi(j)]\right].
\end{eqnarray}
This is the probability that for all $j,i$ 
the measurement $\varpi(j)$ on the $i$ component of the system, will give the same result with both density matrices $\rho$ and $\sigma$.
In other words, for all $i,j$ we perform the measurement $\varpi(j)$ on the $i$ component of a system from an ensemble of systems described by the density matrix $\rho$,
and also on a system from an ensemble of systems described by the density matrix $\sigma$.
We repeat that pair of experiments many times, and $(q_\rho,q_\sigma)$ is the percentage of times that the two experiments give the same result.
In particular
\begin{eqnarray}
(q_\rho,q_\rho)=\prod _{i=0}^{d-1}\left[\sum_{j=0}^{d-1}[{\rm Tr}({\breve \rho}_i\varpi(j))]^2\right].
\end{eqnarray}
Using Eq.(\ref{RRR}) we see that for the states in Eq.(\ref{kkk}) 
\begin{eqnarray}
(q_f,q_g)=(M_f,M_g)=\delta(f,g).
\end{eqnarray}

\end{itemize}

\subsection{Quantum uncertainties in multipartite systems in terms of the Gini index}

Using Eqs.(\ref{G1}), (\ref{G1V}) with the probabilities in Eqs.(\ref{333}),(\ref{333a}) we calculate 
the  Gini vector ${\cal G}(\rho)=({\cal G}_0,...,{\cal G}_{d-1})$ and  the locally dual  Gini vector
${\widetilde {\cal G}}(\rho)={\cal G}(F_L^\dagger\rho F_L)=({\widetilde {\cal G}}_0,...,{\widetilde {\cal G}}_{d-1})$.
Then  Eqs(\ref{Q1}),(\ref{Q2}) show that
\begin{eqnarray}\label{UN1}
\Delta_i(\rho)\ge \eta _d>0;\;\;\;\Delta_i(\rho)=2\frac{d-1}{d+1}-[{\cal G}_i(\rho)+{\widetilde {\cal G}}_i(\rho)],
\end{eqnarray}
This expresses local uncertainty relations in each of the components of this system.

From the joint probabilities $q_\rho(f)$ in Eq.(\ref{RT1}) we can calculate the total Gini index ${\cal G}_T(\rho)$ using Eq.(\ref{G2}). Also 
from the locally dual joint probabilities in Eq.(\ref{RT2}) 
we calculate the locally dual total Gini index ${\widetilde {\cal G}}_T(\rho)$.
The uncertainty relation in Eqs(\ref{Q1}),(\ref{Q2}) becomes here
\begin{eqnarray}\label{UN2}
&&\Delta_T(\rho)\ge \eta _{d^d}>0;\;\;\;\Delta_T(\rho)=2\frac{d^d-1}{d^d+1}-[{\cal G}_T(\rho)+{\widetilde {\cal G}}_T(\rho)]\nonumber\\
&&\eta _{d^d}=2\frac{d^d-1}{d^d+1}-\sup _{\rho }[{\cal G}_T(\rho)+{\widetilde {\cal G}}_T(\rho)].
\end{eqnarray}

\subsection{Example}\label{ex5}

We consider a bipartite system of two qubits and the density matrices 
\begin{eqnarray}\label{cde}
&&\rho=\ket{u}\bra{u};\;\;\;\ket{u}=a\ket{0,0}+b\ket{1,1};\;\;\;|a|^2+|b|^2=1;\;\;\;|a|<|b|\nonumber\\
&&\sigma=\ket{v}\bra{v};\;\;\;\ket{v}=c\ket{0,0}+d\ket{1,0}+e\ket{0,1};\;\;\;|c|^2+|d|^2+|e|^2=1;\;\;\;|e|<|d|<|c|.
\end{eqnarray}
The reduced density matrices are
\begin{eqnarray}
&&{\breve \rho}_0=|a|^2\ket{0}\bra{0}+|b|^2\ket{1}\bra{1}\nonumber\\
&&{\breve \rho}_1=|a|^2\ket{0}\bra{0}+|b|^2\ket{1}\bra{1}\nonumber\\
&&{\breve \sigma}_0=(|c|^2+|d|^2)\ket{0}\bra{0}+|e|^2\ket{1}\bra{1}+ce^*\ket{0}\bra{1}+c^*e\ket{1}\bra{0}\nonumber\\
&&{\breve \sigma}_1=(|c|^2+|e|^2)\ket{0}\bra{0}+|d|^2\ket{1}\bra{1}+cd^*\ket{0}\bra{1}+c^*d\ket{1}\bra{0},
\end{eqnarray}
and from them we find the row Markov matrices of probabilities
\begin{eqnarray}
q_\rho=\begin{pmatrix}
|a|^2&|b|^2\\
|a|^2&|b|^2\\
\end{pmatrix};\;\;\;
q_\sigma=\begin{pmatrix}
|c|^2+|d|^2&|e|^2\\
|c|^2+|e|^2&|d|^2\\
\end{pmatrix}
\end{eqnarray}
The joint probabilities  $q(f(0),f(1))$, the products of probabilities ${\mathfrak M}(f(0),f(1))$,  and the correlations ${\cal C}(f(0),f(1))$
for the density matrices $\rho, \sigma$ are given in table \ref{t2}.

We also used Eqs.(\ref{G1}),(\ref{G2}) to calculate the  Gini vectors
\begin{eqnarray}
{\cal G}(\rho)=\left (\frac{1-2|a|^2}{3}, \frac{1-2|a|^2}{3}\right);\;\;\;{\cal G}(\sigma)=\left (\frac{1-2|e|^2}{3}, \frac{1-2|d|^2}{3}\right)
\end{eqnarray}
and the total Gini indices
\begin{eqnarray}
{\cal G}_T(\rho)=\frac{3-2|a|^2}{5};\;\;\;{\cal G}_T(\sigma)=\frac{3-2|e|^2-|d|^2}{5}.
\end{eqnarray}
It is seen that for small values of $|a|$ the probability vectors $q_\rho(i,j)$ (with fixed arbitrary $i$) are very sparse, and also the joint probability vector $q_\rho(f)$ is very sparse.
Similarly for small values of $|e|$, $|d|$,  the probability vectors $q_\sigma(i,j)$ (with fixed arbitrary $i$) are very sparse, and also the joint probability vector $q_\sigma(f)$ is very sparse.

Furthermore we calculated the
\begin{eqnarray}
&&(q_\rho, q_\sigma)=|a|^2(|ac|^2+|b|^2-|bc|^2)+|ed|^2(|a|^2-|b|^2)^2,\nonumber\\
&&(q_\rho, q_\rho)=(|a|^4+|b|^4)^2,\nonumber\\
&&(q_\sigma, q_\sigma)=[(|c|^2+|d|^2)^2+|e|^4][(|c|^2+|e|^2)^2+|d|^4].
\end{eqnarray}
Some special cases are:
\begin{eqnarray}
&&|a|=|e|=0\;\rightarrow\;(q_\rho, q_\sigma)=0\nonumber\\
&&|a|=1,\;\;|b|=0\;\rightarrow\;(q_\rho, q_\rho)=1\nonumber\\
&&|a|=|b|=\frac{1}{2}\;\rightarrow\;(q_\rho, q_\rho)=\frac{1}{4}\nonumber\\
&&|c|=|d|=|e|=\frac{1}{3}\;\rightarrow\;(q_\sigma, q_\sigma)=\frac{25}{81}.
\end{eqnarray}
In the case $|a|=|e|=0$ the measurements $\varpi(j)$ on the various components of the system described by $\rho$ will never give the same results as the corresponding measurements 
on the system described by $\sigma$.
In the case $|a|=1$ and $|b|=0$, these measurement on two states from the ensemble described by $\rho$ will always give the same result, etc.

\section{Global Fourier transforms}\label{GF}

The quantities in the previous section do not depend on entangling off-diagonal elements in the density matrix
 (e.g., the quantities in the example in section \ref{ex5} depend only on $|a|^2, |b|^2,...$).
This motivates the introduction of the globally dual quantities in this section based on a `global Fourier transform'. 
The quantities in this section depend on off-diagonal elements that entangle the various components of the system. 
We note here that the non-diagonal elements are necessary but not sufficient requirement for entanglement.

We consider a bijective map between $({\mathbb Z}_d)^d={\mathbb Z}_d\times ...\times {\mathbb Z}_d$ and ${\mathbb Z}_{d^d}$ as follows.
We first take each $j_r$ in the `period' $0\le j_r\le d-1$ and $\widehat j$ in the `period'  $0\le \widehat j\le d^d-1$, and introduce the bijective map
\begin{eqnarray}\label{sss}
j=(j_0,...,j_{d-1})\;\leftrightarrow\;\widehat j=j_0+j_1d+...+j_{d-1}d^{d-1}.
\end{eqnarray}
We then take each $j_r$ modulo $d$ and the $\widehat j$ modulo $d^d$, and we get a bijective map from $({\mathbb Z}_d)^d$ to ${\mathbb Z}_{d^d}$.
We note that if $j=(j_0,...,j_{d-1})$ and $k=(k_0,...,k_{d-1})$ then
\begin{eqnarray}
\widehat j\widehat k=j_0k_0+d(j_0k_1+j_1k_0)+...+d^{d-1}(j_0k_{d-1}+...+j_{d-1}k_0)
\end{eqnarray}

The Hilbert space ${\mathfrak H}$ is isomorphic to $H_{d^d}$ (a $d^d$-dimensional Hilbert space describing systems with variables in ${\mathbb Z}_{d^d}$). But the $({\mathbb Z}_d)^d$  as a ring (with addition and multiplication componentwise), is not isomorphic to the ring ${\mathbb Z}_{d^d}$ because $\widehat j+\widehat k\ne \widehat {j+k}$ and $\widehat j\cdot\widehat k\ne \widehat {j\cdot k}$.
For example, in the case $d=3$ we get
\begin{eqnarray}
&&\widehat{(2,1,2)}+\widehat{(1,1,0)}=27;\;\;\;\widehat{(2,1,2)+(1,1,0)}=\widehat{(0,2,2)}=24\nonumber\\
&&\widehat{(2,1,2)}\cdot\widehat{(1,1,0)}=23\cdot4=92=11;\;\;\;\widehat{(2,1,2)\cdot (1,1,0)}=\widehat{(2,1,0)}=5.
\end{eqnarray}
It is seen that $[{\mathbb Z}_{3}]^3$ is non-isomorphic to ${\mathbb Z}_{27}$.
Consequently our `local formalism' in the phase space $({\mathbb Z}_d\times{\mathbb Z}_d)^d$ is different from our `global formalism' in the phase space ${\mathbb Z}_{d^d}\times{\mathbb Z}_{d^d}$.
In this paper we introduce Fourier transforms in both cases.

We  introduce the global Fourier transform in ${\mathfrak H}$:
\begin{eqnarray}
&&F_G=\frac{1}{\sqrt{d^d}}\sum _{{\widehat j},{\widehat k}}\omega_{d^d}(\widehat j\widehat k) \ket{j_0,...,j_{d-1}}\bra{k_0,...,k_{d-1}};\;\;\;\omega_{d^d}(\alpha)=\exp \left (i\frac{2\pi\alpha}{d^d}\right);\;\;\;\alpha\in{\mathbb Z}_{d^d}\nonumber\\
&&\omega_{d^d}(\widehat j\widehat k)=\omega_{d^d}[j_0k_0+d(j_0k_1+j_1k_0)+...+d^{d-1}(j_0k_{d-1}+...+j_{d-1}k_0)]\nonumber\\
&&F_G^4={\bf 1};\;\;\;F_GF_G^{\dagger}={\bf 1};\;\;\;F_G\ne F_L.
\end{eqnarray}
The index $G$ in the notation stands for global.

Using the fact that
\begin{eqnarray}
\frac{1}{{d^d}}\sum _{ \widehat k}\omega_{d^d}[(\widehat j+\widehat \ell )\widehat k] =\delta (\widehat j+\widehat \ell ,0);\;\;\;\widehat j, \widehat \ell\in {\mathbb Z}_{d^d}
\end{eqnarray}
in conjunction with the fact that ${\widehat j}+{\widehat \ell}=0$ (${\rm mod}(d^d)$) implies $j_r+\ell _r=0$ (${\rm mod}(d)$),
we prove that $F_G^2$ is the parity operator
\begin{eqnarray}
&&F_G^2=\frac{1}{{d^d}}\sum _{\widehat j, \widehat k, \widehat \ell}\omega_{d^d}[(\widehat j+\widehat \ell )\widehat k] \ket{j_0,...,j_{d-1}}\bra{\ell_0,...,\ell_{d-1}}=
\sum_{j_0,...,j_{d-1}}\ket{j_0,...,j_{d-1}}\bra{-j_0,...,-j_{d-1}}.
\end{eqnarray}
Also using Eq.(\ref{parity}) we prove that $F_L^2$ is the parity operator. Therefore
\begin{eqnarray}
F_G^2=F_L^2.
\end{eqnarray}
Extra care is required in practical calculations, with the modular arithmetic of the indices.

$F_G$ are global transformations in the sense that they cannot be written as $U_0\otimes ...\otimes U_{d-1}$ where $U_r$ are local unitary transformations.
Acting with $F_G$ on the basis $\ket{j_0,...,j_{d-1}}$ we get the globally dual basis
\begin{eqnarray}
&&\ket{j_0,...,j_{d-1}}_{\rm G}=F_G\ket{j_0,...,j_{d-1}}
=\frac{1}{\sqrt{d^d}}\bigotimes _{r=0}^{d-1}\left [\sum _{k_r=0}^{d-1}\omega_{d^d}[(j_0d^r+...+j_{d-r-1}d^{d-1})k_r]\ket{k_r}\right ]
\end{eqnarray}
The states $\ket{j_0,...,j_{d-1}}_{\rm G}$  are `global', in the sense that the coefficients $\omega_{d^d}[(j_0d^r+...+j_{d-r-1}d^{d-1})k_r]$ of the vectors in the $r$-component depend on all $j_0,...,j_{d-1}$.
Information from all components $j_0,...,j_{d-1}$ is needed to determine each of these coefficients.
In the local Fourier transform of Eq.(\ref{103}), the coefficients $\omega_{d}(j_rk_r)$ in the $r$-component depend only on $j_r$.

We will use the term `globally Fourier transformed factorisable states' for the states $F_G\ket{s}$ where $\ket{s}$ are factorisable states.
The example below shows that a globally Fourier transformed factorisable state is in general an entangled state, but in special cases it can be factorisable.

It is easily seen that
\begin{eqnarray}
&&_L\langle\ell _0,...,\ell_{d-1}\ket{j_0,...,j_{d-1}}_{\rm G}=\frac{1}{d^d}\sum_{\widehat k}\omega_{d^d}(\widehat j \widehat k)\omega_d[-(\ell_0k_0+...+\ell_{d-1}k_{d-1})]\nonumber\\
&&|\langle\ell _0,...,\ell_{d-1}\ket{j_0,...,j_{d-1}}_{\rm G}|^2=|\langle\ell _0,...,\ell_{d-1}\ket{j_0,...,j_{d-1}}_{\rm L}|^2=\frac{1}{d^d}
\end{eqnarray}

\begin{example}
In the case  $d=3$, we act with $F_G$ on the factorisable state $\ket{j_0,j_1,j_2}$ and we get
\begin{eqnarray}\label{dd}
&&\ket{j_0,j_1,j_2}_{\rm G}
=\frac{1}{\sqrt{27}}\sum _{k_0,k_1,k_2}\omega_{27}[j_0k_0+3(j_1k_0+j_0k_1)+9(j_2k_0+j_1k_1+j_0k_2)]\ket{k_0,k_1,k_2}\nonumber\\
&&=\frac{1}{\sqrt{27}}\sum _{k_0}\omega_{27}[j_0k_0+3j_1k_0+9j_2k_0]\ket{k_0}\otimes
\sum _{k_1}\omega_{27}(3j_0k_0+9j_1k_1)\ket{k_1}\otimes
\sum _{k_2}\omega_{27}(9j_0k_2)\ket{k_2}
\end{eqnarray}
This is an example of a globally Fourier transformed factorisable state which is factorisable.

Also we act with $F_G$ on the factorisable state
\begin{eqnarray}
\ket{s}=a\ket{j_0,j_1,j_2}+b\ket{r_0,j_1,j_2};\;\;\;|a|^2+|b|^2=1,
\end{eqnarray}
and we get the following globally Fourier transformed factorisable state which is entangled (when $a,b\ne 0$):
\begin{eqnarray}\label{dd}
F_G\ket{s}&=&\frac{a}{\sqrt{27}}\sum _{k_0}\omega_{27}[j_0k_0+3j_1k_0+9j_2k_0]\ket{k_0}\otimes
\sum _{k_1}\omega_{27}(3j_0k_0+9j_1k_1)\ket{k_1}\otimes
\sum _{k_2}\omega_{27}(9j_0k_2)\ket{k_2}\nonumber\\
&+&\frac{b}{\sqrt{27}}\sum _{k_0}\omega_{27}[r_0k_0+3j_1k_0+9j_2k_0]\ket{k_0}\otimes
\sum _{k_1}\omega_{27}(3r_0k_0+9j_1k_1)\ket{k_1}\otimes
\sum _{k_2}\omega_{27}(9r_0k_2)\ket{k_2}.
\end{eqnarray}

\end{example}

\subsection{Globally dual quantities}
Acting with $F_G$ on the quantities in section \ref{Q} we get their duals denoted with a `hat'.
We introduce them because they depend on the entangling off-diagonal elements, which do not enter in the quantities in section \ref{Q}.
 For example, we introduce the projectors
\begin{eqnarray}
\widehat \Pi(i,j)=F_G\Pi(i,j)F_G^{\dagger};\;\;\;
\widehat \Pi(f)=F_G\Pi[f(0),...,f(d-1)]F_G^{\dagger}=\prod_i\widehat \Pi[i,f(i)],
\end{eqnarray}
the row Markov matrix $\widehat q_\rho$ with the globally dual probabilities
\begin{eqnarray}
\widehat q_\rho(i,j)={\rm Tr}[\rho\widehat \Pi(i,j)]={\rm Tr}[F_G^\dagger\rho F_G \Pi(i,j)]=q_{F_G^\dagger\rho F_G}(i,j),
\end{eqnarray}
the globally dual products of probabilities
\begin{eqnarray}
{\widehat {\mathfrak M}}_\rho(f)=\prod _i \widehat q_\rho(i,f(i))={\mathfrak M}_{F_G^\dagger\rho F_G}(f),
\end{eqnarray}
and the globally dual joint probabilities
\begin{eqnarray}\label{qaz}
\widehat q_\rho(f)={\rm Tr}[\rho  \widehat  \Pi(f)]=q_{F_G^\dagger\rho F_G}(f).
\end{eqnarray}
So the globally dual probabilities for the density matrix $\rho$, are the probabilities of the Fourier transformed density matrix $F_G^\dagger\rho F_G$.
Using the globally dual probabilities  we  can calculate the globally dual correlation coefficients
${\widehat {\cal C}}_\rho (f)$ using Eq.(\ref{cor}), and we see that
\begin{eqnarray}
{\widehat {\cal C}}_\rho (f)={\cal C}_{F_G^\dagger\rho F_G} (f).
\end{eqnarray}
We also calculate the globally dual Gini vector ${\widehat{\cal G}}(\rho)=({\widehat {\cal G}}_0,...,{\widehat {\cal G}}_{d-1})$ using Eqs.(\ref{G1}), (\ref{G1V}) and the   globally dual total Gini index $\widehat {\cal G}_T(\rho)$ using Eq.(\ref{G2}),
and we get
\begin{eqnarray}
{\widehat{\cal G}}(\rho)={\cal G}(F_G^\dagger\rho F_G);\;\;\;\widehat {\cal G}_T(\rho)={\cal G}_T(F_G^\dagger\rho F_G).
\end{eqnarray}
Since ${\widehat{\cal G}}_i(\rho)$ and ${\cal G}_i(\rho)$ are related through a Fourier transform, we have the uncertainty relations
\begin{eqnarray}\label{U1}
&&D_i(\rho)\ge {\widehat \eta} _d>0;\;\;\;D_i(\rho)=2\frac{d-1}{d+1}-[{\cal G}_i(\rho)+{\widehat {\cal G}}_i(\rho)]\nonumber\\
&&{\widehat \eta} _{d}=2\frac{d-1}{d+1}-\sup _{\rho }[{\cal G}_i(\rho)+{\widehat {\cal G}}_i(\rho)].
\end{eqnarray}
All the component systems are the same, and therefore ${\widehat \eta} _{d}$ does not depend on the index $i$.
Also the ${\widehat{\cal G}}_T(\rho)$ and ${\cal G}_T(\rho)$ are related through a Fourier transform, and we have the uncertainty relation
\begin{eqnarray}\label{U2}
&&D_T(\rho)\ge {\widehat \eta} _{d^d}>0;\;\;\;D_T(\rho)=2\frac{d^d-1}{d^d+1}-[{\cal G}_T(\rho)+{\widehat{\cal G}}_T(\rho)]\nonumber\\
&&{\widehat \eta} _{d^d}=2\frac{d^d-1}{d^d+1}-\sup _{\rho }[{\cal G}_T(\rho)+{\widehat {\cal G}}_T(\rho)].
\end{eqnarray}

\begin{example}
We consider a tripartite system of qutrits ($d=3)$. In the Hilbert space $H_3\otimes H_3\otimes H_3$
we consider the density matrices
\begin{eqnarray}
&&\rho=\ket{r}\bra {r};\;\;\;\ket{r}=\frac{1}{\sqrt{3}}[\ket{0,0,0}+\ket{1,1,0}+\ket{2,2,1}]\nonumber\\
&&\sigma=\frac{1}{3}[\ket{0,0,0}\bra{0,0,0}+\ket{1,1,0}\bra{1,1,0}+\ket{2,2,1}\bra{2,2,1}],
\end{eqnarray}
Using the global Fourier transform described in Eq.(\ref{dd}), we calculated
the row Markov matrix $\widehat q_\rho$ that contains the globally dual probabilities:
\begin{eqnarray}
\widehat q_\rho=\begin{pmatrix}
0.325&0.422&0.253\\
0.394&0.322&0.284\\
0.333&0.333&0.333
\end{pmatrix};\;\;\;
\widehat q_\sigma=\begin{pmatrix}
0.333&0.333&0.333\\
0.333&0.333&0.333\\
0.333&0.333&0.333
\end{pmatrix}
\end{eqnarray}
We then put the probabilities in each row in ascending order, and  we calculated the globally dual Gini vectors
\begin{eqnarray}
{\widehat{\cal G}}(\rho)=(0.085, 0.055,0);\;\;\;{\widehat{\cal G}}(\sigma)=(0,0,0).
\end{eqnarray}
We also calculated the $27$  globally dual joint probabilities in Eq.(\ref{qaz}),  put them in ascending order,
and found the  globally dual total Gini indices
\begin{eqnarray}
\widehat {\cal G}_T(\rho)=0.430;\;\;\;\widehat {\cal G}_T(\sigma)=0.
\end{eqnarray}
The results for $\rho$ are different from the results for $\sigma$, and the difference between these two density matrices are off-diagonal entangling elements.
Therefore the globally dual quantities depend on the off-diagonal entangling elements.

\end{example}

\subsection{Open problems}
In this subsection we mention briefly some open problems.
\begin{itemize}
\item
The `local formalism' in the phase space $({\mathbb Z}_d\times{\mathbb Z}_d)^d$ is different from the `global formalism' in the phase space ${\mathbb Z}_{d^d}\times{\mathbb Z}_{d^d}$.
In a multipartite system where the various parties are distinct physical systems that do not interact with each other, you can argue that the 
phase space $({\mathbb Z}_d\times{\mathbb Z}_d)^d$ is more physical. But if the various parties interact with each other, this is no longer true.
In this paper we introduced Fourier transforms in both cases, and in further work other phase space quantities (e.g., displacement operators, Wigner and Weyl functions, etc) can also be defined for the two cases.
Comparison of the two formalisms for various examples, might shed light into the
relationship between global transformations and entanglement.
\item
We can define `globally Fourier transformed separable states' as $F_G\rho_{\rm sep}F_G^{\dagger}$ where $\rho_{\rm sep}$ is a separable mixed state.
Comparison of the two phase space formalisms for various examples, might shed light into various aspects of the entanglement for mixed states.
\item
In many multipartite systems there is no natural ordering of the various components.
Therefore instead of using Eq.(\ref{sss}) we can first perform a permutation $\pi$ on the various components:
\begin{eqnarray}
j=(j_0,...,j_{d-1})\;\rightarrow\;j_\pi=(j_{\pi(0)},...,j_{\pi(d-1)})\;\rightarrow\;\widehat j_\pi=j_{\pi(0)}+j_{\pi (1)}d+...+j_{\pi(d-1)}d^{d-1}.
\end{eqnarray} 
So there are many global Fourier transforms (one for each permutation).
\item
A more general problem will be to consider 
an $n$-partite system comprised of $n$ components each of which is a qudit with $n\ne d$.
In the language of permutations with repetitions we have sequences of $n$ integers from ${\mathbb Z}_d$, and we deal with non-square matrices.
It is not clear how the material in section 4, can be generalised to non-square matrices.
\item
In this paper a mathematical structure (permutations with repetitions), is transformed with quantum techniques into a more general structure (the fundamental difference is the 
concept of superpositions).
This could be applied to many algebraic structures.

\end{itemize}

\section{Discussion}

We have blended ideas from three different areas.
\begin{itemize}
\item
The first area is row Markov matrices.
 We introduced expansions for row Markov matrices, in terms of matrices related to permutations with repetitions.
We interpreted this in terms of random safes described by the Markov matrices.
In the expansion of Eq.(\ref{ghd}) the coefficients are joint probabilities, and in the expansion in Eq.(\ref{gh}) the coefficients are products of probabilities.
The difference between the two are the correlations in Eq.(\ref{COR}).

\item
The second area is Lorenz values and the Gini index.
We used them to quantify the sparsity of probability vectors.
The properties of Lorenz values have been presented in propositions \ref{L1}, \ref{pro78}, and the properties of the Gini index in propositions \ref{proG1}, \ref{pro56}, \ref{GG3}.
In the context of random safes we introduced  in Eqs(\ref{G1}),(\ref{G1V}) the Gini vector that describes the sparsity 
in the local probability vector for  each of the  integers in the sequence that opens a random safe. 
We also introduced the total Gini index of Eq.(\ref{G2}) that describes the sparsity of the 
joint probabilities.
\item
The third area is multipartite quantum systems.
We viewed them as quantum  permutations with repetitions and as quantum safes, and then used the above two formalisms in a quantum context. 
In section \ref{markov} we presented the Markov matrix formalism in a quantum context.
This led to novel statistical quantities that describe classical and quantum correlations in multipartite quantum systems.
Local Fourier transforms led to locally dual statistical quantities.
Global Fourier transforms led to globally dual statistical quantities which
 depend on off-diagonal elements that entangle the various components of the system.
In Eqs(\ref{UN1}),(\ref{UN2}),(\ref{U1}),(\ref{U2}) we gave uncertainty relations in terms of Gini indices. 

\end{itemize}

The first two parts are related to classical probabilistic multipartite systems (random safes).
The third part is related to quantum multipartite systems (quantum safes).
The work introduces novel methods into multipartite quantum systems.
It also shows that the quantum concept of superposition can be introduced in some mathematical areas (in our case permutations with repetitions) and generalise them into new areas.

\newpage
\begin{table}
\caption{The expansion in Eq.(\ref{gh}) (which assumes independence) for the row Markov matrix in Eq.(\ref{AB1}).
Only $8$ of the terms are assigned non-zero probability and the corresponding matrices $M_f$ together with their joint probabilities ${\mathfrak M}_q(f)$ (which are products of probabilities) for various functions $f\in{\cal F}$, are shown in the first three columns.
A different expansion for the same Markov matrix (in the presence of correlations) is given in Eq.(\ref{kk}) and the corresponding joint probabilities $q(f)$ and correlation coefficients ${\cal C}_q(f)=q(f)-{\mathfrak M}_q(f)$ are shown
in the last two columns.}
\def\arraystretch{2}
\begin{tabular}{|c|c|c||c|c|}\hline
$(f(0),f(1),f(2))$&$M_f$&${\mathfrak M}_q(f)$&$q(f)$&${\cal C}_q(f)$\\\hline
$(0,1,1)$&$\begin{pmatrix}
1&0&0\\
0&1&0\\
0&1&0
\end{pmatrix}$&
$a^2(1-b)$&$0$&$-a^2(1-b)$\\\hline
$(0,1,2)$&$\begin{pmatrix}
1&0&0\\
0&1&0\\
0&0&1
\end{pmatrix}$&
$a^2b$&$a$&
$a-a^2b$\\\hline
$(1,1,1)$&$\begin{pmatrix}
0&1&0\\
0&1&0\\
0&1&0
\end{pmatrix}$&
$a(1-a)(1-b)$&$0$&
$-a(1-a)(1-b)$\\\hline
$(1,1,2)$&$\begin{pmatrix}
0&1&0\\
0&1&0\\
0&0&1
\end{pmatrix}$&
$a(1-a)b$&$0$&
$-a(1-a)b$\\\hline
$(0,2,1)$&$\begin{pmatrix}
1&0&0\\
0&0&1\\
0&1&0
\end{pmatrix}$&
$a(1-a)(1-b)$&$0$&
$-a(1-a)(1-b)$\\\hline
$(0,2,2)$&$\begin{pmatrix}
1&0&0\\
0&0&1\\
0&0&1
\end{pmatrix}$&
$a(1-a)b$&$0$&
$-a(1-a)b$\\\hline
$(1,2,1)$&$\begin{pmatrix}
0&1&0\\
0&0&1\\
0&1&0
\end{pmatrix}$&
$(1-a)^2(1-b)$&$1-b$&
$(1-b)(2a-a^2)$\\\hline
$(1,2,2)$&$\begin{pmatrix}
0&1&0\\
0&0&1\\
0&0&1
\end{pmatrix}$&
$(1-a)^2b$&$b-a$&
$2ab-a-a^2b$\\\hline
\end{tabular} \label{t1}
\end{table}

\begin{table}
\caption{The joint probabilities  $q(f(0),f(1))$, the products of probabilities ${\mathfrak M}(f(0),f(1))$,  and the correlations ${\cal C}(f(0),f(1))$
for the density matrices $\rho, \sigma$ in Eq.(\ref{cde}).}
\def\arraystretch{2}
\begin{tabular}{|c||c|c|c||c|c|c|}\hline
$(f(0),f(1))$&$q_\rho(f)$&${\mathfrak M}_\rho(f)$&${\cal C}_\rho(f)$&$q_\sigma(f)$&${\mathfrak M}_\sigma(f)$&${\cal C}_\sigma(f)$\\\hline
$(0,0)$&$|a|^2$&$|a|^4$&$|ab|^2$&$|c|^2$&$|c|^2+|ed|^2$&$-|de|^2$\\\hline
$(0,1)$&$0$&$|ab|^2$&$-|ab|^2$&$|e|^2$&$|d|^2-|ed|^2$&$|e|^2-|d|^2+|ed|^2$\\\hline
$(1,0)$&$0$&$|ab|^2$&$-|ab|^2$&$|d|^2$&$|e|^2-|ed|^2$&$|d|^2-|e|^2+|ed|^2$\\\hline
$(1,1)$&$|b|^2$&$|b|^4$&$|ab|^2$&$0$&$|e|^2|d|^2$&$-|ed|^2$\\\hline
\end{tabular} \label{t2}
\end{table}

\end{document}